\documentclass[a4paper,conference]{IEEEtran}
\usepackage{amsthm}
\newtheorem{theorem}{Theorem}
\theoremstyle{definition}
\newtheorem{example}{Example}
\newtheorem{remark}{Remark}
\usepackage[utf8]{inputenc}
\usepackage{subfigure}
\usepackage{url}
\usepackage{algorithm}
\usepackage[noend]{algpseudocode}
\usepackage{bm}
\usepackage{amsfonts}
\usepackage{amsmath}
\usepackage{bigdelim}
\usepackage{multirow}
\usepackage{subfigure}
\usepackage{url}
\usepackage{algorithm}
\usepackage[noend]{algpseudocode}
\usepackage{bm}
\usepackage{amsfonts}
\usepackage{cite}

\ifCLASSINFOpdf
   \usepackage[pdftex]{graphicx}
\else
   \usepackage[dvips]{graphicx}
\fi

\title{A study on the ideal magnitude and phase of reconstructed point targets in SAR imaging}
\author{\IEEEauthorblockN{
Guanying Sun,
Carey Rappaport\IEEEauthorrefmark{1}
}                                    
\IEEEauthorblockA{\IEEEauthorrefmark{1}Department of Electrical and Computer Engineering, Northeastern University, Boston, USA, rappaport@coe.neu.edu } 
}

\begin{document}
\maketitle

\begin{abstract}
In this paper, the magnitude and phase of the reconstructed point targets in SAR imaging are studied quantitatively by using inverse crime. Two scenarios, 
one with single point target in the imaging area and the other with two point targets,
are considered. 
The theorems on the magnitude and phase are established and proved for each scenario. In addition, several numerical examples are presented and the numerical results show that they agree with the corresponding theorems. 
This study is useful for appreciating the limitations of formulating inversion algorithms based on simplistic point target building blocks.
\end{abstract}


\IEEEpeerreviewmaketitle

\section{Introduction}
Numerous SAR imaging techniques have been developed over the past several decades \cite{
Carrara,Soumekh,Mata}. In this paper we consider a noiseless, linear imaging system $AX = b$, where $A$ is the sensing matrix, $b$ is the data received by the receiver, and $X$ is the unknown reflectivity vector, i.e, the reconstructed solution. 
Here, b is either measured or numerically calculated.
The sensing matrix $A$ used in this paper was discussed in \cite{Galia} (Ghazi, 2017, p. 26-27), with
elements the form of $e^{-jkL}$, where $k$ is the wave number, $L$ is  the path length from the transmitter to the scattering point, then to the receiver. In this paper, we only consider the configuration of single transmitter and single receiver with multiple frequencies. 

The inverse method we use in this work is the adjoint method: 
$X \approx A^* b$, where $A^*$ is the conjugate transpose of $A$. We calculate $b$ by multiplying sensing matrix $A$ by the exact solution, that is, committing an inverse crime. The inverse crime arises if an inverse problem is solved using a specific method and then tested by solving the forward problem with the same or nearly the same method or vice versa \cite{ColtonKress,Potthast, Hansen}. 
Our work in this paper presents an angle of view which helps understanding the point targets reconstruction in SAR imaging and appreciating the limitations of this simplistic model. 

It is worth mentioning that, since phase is very sensitive to noisy data, it is difficult to make use of phase itself in the signal processing of SAR. For comparison, the variance of phase is more useful in reality. For example, it can be used in reducing the side lobes in acoustic and SAR imaging \cite{Camacho, Baccouche, EuCAP2020}.

The rest of this paper is organized as follows. In Section II, we discuss the reconstruction of single point target and prove a theorem on the magnitude and phase of the reconstructed value. In Section III, the case with two point targets is studied, and four theorems on the magnitude and phase are proposed. Numerical examples are provided in Section IV to verify the theorems proposed in Sections II and III. At last, we draw the main conclusions in Section V.

\section{Reconstruction of single point target}
Assume $M$ frequencies are used throughout this paper, then we have $M$ wave numbers $k_i, i=1,2,...,M$.
We also assume that there are $N$ imaging pixels in the imaging area and the $n$-th pixel locates at $x_n$.
Let $L_n$ represent the total path length from the transmitter to imaging pixel $x_n$, then to the receiver. 
The sensing matrix \cite{Galia} is then given by:
\begin{align}
 A & = 
 \begin{bmatrix}
 e^{-jk_1L_1} & \hspace{-2mm} e^{-jk_1L_2}& \hspace{-3mm} \cdots & \hspace{-3mm}  e^{-jk_1L_N} \\ 
 e^{-jk_2L_1} &\hspace{-2mm} e^{-jk_2L_2} & \hspace{-3mm} \cdots & \hspace{-3mm} e^{-jk_2L_N} \\
 \vdots & \vdots & \hspace{-3mm}\vdots & \vdots \\
e^{-jk_ML_1} &\hspace{-2mm} e^{-jk_ML_2} &\hspace{-3mm} \cdots &\hspace{-3mm} e^{-jk_ML_N} \hspace{-1mm}
\end{bmatrix}_{\hspace{-1mm} M \times N} \nonumber\\
& := [A_1, A_2, \cdots, A_N ],  
\end{align}
where $A_i$ is the $i$-th column of matrix $A$. 
The imaging equation can be written as:
\begin{equation}
AX = b,
\end{equation}
where $b$ is the set of received data,
and $X$ is the set of the reflectivities of all pixels in the imaging area.

In this section, we discuss the reconstruction of single point target. The geometry of the imaging system is shown in Fig.~\ref{fig:Imaging geometry 1}.  The point target at arbitrary point $x_p$  is illuminated by the transmitter (Tx) and the reflected signal is received by the receiver (Rx). 
\begin{figure}[ht]
\centering
\vspace{-4mm}
\includegraphics[width=0.5\columnwidth]{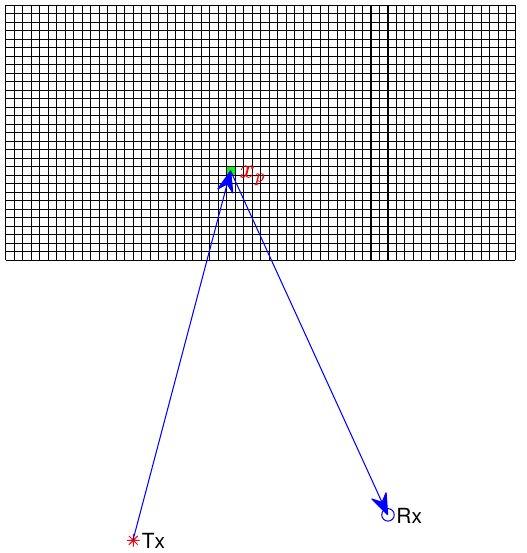}
\caption{Imaging geometry for single point target. }
\label{fig:Imaging geometry 1}
\end{figure}

Next, we prove the following theorem about the magnitude and phase of the reconstructed single point target. 
\begin{theorem}
Assume that single point target locates at arbitrary point $x_p$ of the imaging area, then the magnitude and phase of the reconstructed point target are $M$ and $0$, respectively. 
\end{theorem}
\begin{proof}
The reconstruction method can be described as: $\tilde{X} = A^*b$, where $\tilde{X}$ is the reconstructed solution, and $A^*$ is the conjugate transpose of $A$, i.e.,
\begin{align}
 \mathbf{A^*} = \begin{bmatrix}
 e^{jk_1L_1} &  e^{jk_2L_1} & \cdots &  e^{jk_ML_1} \\
 e^{jk_1L_2} & e^{jk_2L_2} & \cdots & e^{jk_ML_2} \\
 \vdots & \vdots & \vdots & \vdots \\
e^{jk_1L_N} & e^{jk_2L_N} & \cdots & e^{jk_ML_N}  \\
\end{bmatrix}_{N \times M}  
:=\begin{bmatrix}
A_1^* \\
A_2^*\\ 
\vdots\\
A_N^* 
\end{bmatrix},
\end{align}
where $A_i^*$ is the conjugate transpose of $A_i$. 

Since the point target is at $x_p$, the exact solution $X^o$ is a $N \times 1$ unit vector with the $p$-th element being $1$. By inverse crime, the received signal at the receiver $R_x$ is calculated by $b = AX^o = A_p$, where $A_p$ is the $p$-th column of $A$.
Then, we can obtain
\begin{equation}
  \tilde{X} = A^*b = \begin{bmatrix}
A_1^* \\
A_2^*\\ 
\vdots\\
A_N^* 
\end{bmatrix} A_p = 
\begin{bmatrix}
A_1^* A_p\\
A_2^* A_p\\ 
\vdots\\
A_p^*A_p\\
\vdots\\
A_N^* A_p
\end{bmatrix}, \label{eqn4}
\end{equation}
where block multiplication is utilized in the last step. 

The $p$-th element of $\tilde{X}$, i.e., $A_p^*A_p$, is the reconstructed value of the point target at $x_p$. 
Since 
\[
A_p^*A_p = [e^{jk_1L_p}, e^{jk_2L_p}, \cdots, e^{jk_ML_p} ]\begin{bmatrix}
e^{-jk_1L_p}\\
e^{-jk_2L_p}\\ 
\vdots\\
e^{-jk_ML_p}
\end{bmatrix}=M,
\]
the magnitude and phase of the
reconstructed point target are $M$ and $0$, respectively.
\end{proof}
From (\ref{eqn4}), it can be easily seen that the maximum magnitude of the pixels in the imaging area is $M$, thus the reconstructed single point target has the maximum magnitude, corresponding to the highest intensity in the imaging area.

\section{Reconstruction of two point targets}
In this section we discuss the reconstruction of two point targets.
The geometry of the imaging system is shown in Fig.~\ref{fig:Imaging geometry 2}. The frequencies and the partition of the imaging area are the same as in Section II. The point targets $x_p$ and $x_q$  are illuminated by transmitter Tx and the reflected signals are received by the receiver Rx. The interference between the two point targets are neglected. 
By the definition of $L_n$ in Section II, $L_p$ and $L_q$ are the path lengths from the transmitter to the corresponding point target, then to the receiver. 
\begin{figure}[ht]
\centering
\includegraphics[width=0.5\columnwidth]{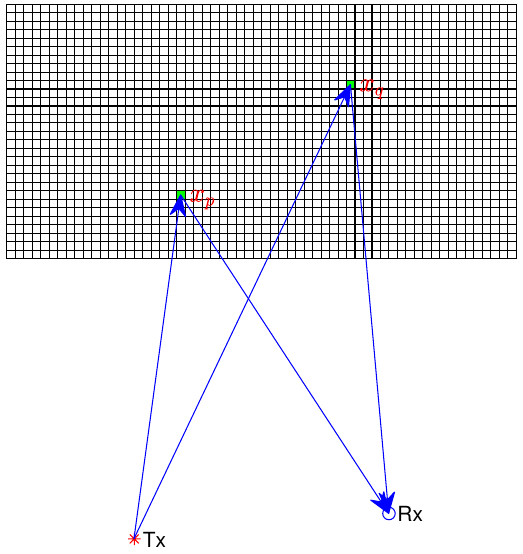}
\caption{Imaging geometry for two point targets. }\vspace{-5mm}
\label{fig:Imaging geometry 2}
\end{figure}

We now establish the following theorem regarding the magnitude and phase of the two reconstructed point targets. 
\begin{theorem}
Assume that two point targets locate at arbitrary points $x_p$ and $x_q$ of the imaging area, then the two reconstructed point targets have the same magnitude and opposite phase.
\end{theorem}
\begin{proof}
Since the two point targets are at  $x_p$ and $x_q$, the exact solution is $X^o=[0,...,0,1,0,...,0,1,0,...0]^T$, where  the $p$-th and $q$-th elements are $1$ and the others are $0$. By inverse crime, the response at the receiver is calculated as $b = AX^o = A_p+A_q$, where $A_p$ and $A_q$ are the $p$-th and $q$-th columns of A, respectively. 
Thus we can obtain the reconstructed solution
\begin{equation}
\tilde{X}=A^*b = \begin{bmatrix}
A_1^* \\
A_2^*\\ 
\vdots\\
A_N^* 
\end{bmatrix}b=
\begin{bmatrix}
A_1^* (A_p+A_q)\\
A_2^* (A_p+A_q)\\ 
\vdots\\
A_N^* (A_p+A_q)
\end{bmatrix},
\end{equation}
where block multiplication is applied in obtaining the last equation. 

The $p$-th and $q$-th elements of $\tilde{X}$, i.e.,  $A_p^*(A_p+A_q)$ and $A_q^*(A_p+A_q)$, are the reconstructed values of the two point targets.
It is obvious that $A_p^*A_p=M$ and $A_q^*A_q=M$.
In addition, we have
\begin{align}
A_p^*A_q &= [e^{jk_1L_p}, e^{jk_2L_p},\cdots, e^{jk_ML_p} ]\begin{bmatrix}
e^{-jk_1L_q}\\ e^{-jk_2L_p}\\
\vdots\\
e^{-jk_ML_q}
\end{bmatrix} \nonumber\\
&= \sum_{i=1}^M e^{jk_i (L_p-L_q)},
\end{align}
and $A_q^*A_p=\sum_{i=1}^M e^{jk_i (L_q-L_p)}.
$
Therefore, the $p$-th element of $\tilde{X}$ is
\vspace{-4mm}
\begin{align}
& \hspace{5mm} A_p^*(A_p+A_q)
= M+\sum_{i=1}^M e^{jk_i (L_p-L_q)}\nonumber\\
&= \{M+\sum_{i=1}^M cos(k_i(L_p-L_q))\}+j\sum_{i=1}^M sin(k_i(L_p-L_q)), \label{eqn7}
\end{align}
and
the $q$-th element of $\tilde{X}$ is
\begin{align}
&\hspace{5mm}A_q^*(A_p+A_q)=M+\sum_{i=1}^M e^{jk_i (L_q-L_p)}\nonumber\\
&= \{M+\sum_{i=1}^M cos(k_i(L_p-L_q))\}
-j\sum_{i=1}^M sin(k_i(L_p-L_q))). \label{eqn8}
\end{align}
From (\ref{eqn7}) and (\ref{eqn8}), it can be seen that the two values are conjugate of each other, thus the two reconstructed point targets have the same magnitude and opposite phase. 
\end{proof}

From (\ref{eqn7}) and (\ref{eqn8}), it can be seen that the maximum potential magnitude of the two point targets is $2M$. In the next theorem, we present the condition under which the magnitude is the maximum, i.e., equal to $2M$.
\begin{theorem}
Assume that two point targets locate at arbitrary points $x_p$ and $x_q$ of the imaging area, then the magnitude of the two reconstructed point targets is the maximum ($=2M$) if and only if  $k_i(L_p-L_q)$ is a multiple of $2\pi$, for any $i=1,2,...,M$.
\end{theorem}
\begin{proof}
Define $\Phi_i=k_i(L_p-L_q)$.\\
$\Leftarrow$:
If $\Phi_i$ is a multiple of $2\pi $, then 
from (\ref{eqn7}) we have 
$A_p^*(A_p+A_q) = M+\sum_{i=1}^M e^{j\Phi_i}=M+M=2M$.\\
$\Rightarrow$: We have $|M+\sum_{i=1}^M e^{j\Phi_i} |=2M$. 
Since
\begin{align} 
& \hspace{5mm} M+\sum_{i=1}^M e^{j\Phi_i} = \sum_{i=1}^M (1+e^{j\Phi_i})\nonumber\\
&= \sum_{i=1}^M e^{j\frac{\Phi_i}{2}}(e^{-j\frac{\Phi_i}{2}}+e^{j\frac{\Phi_i}{2}})
=2 \sum_{i=1}^M e^{j\frac{\Phi_i}{2}} cos(\frac{\Phi_i}{2}),
 \label{eqn9}
\end{align}
 it can be obtained that
 $M=|\sum_{i=1}^M e^{j\frac{\Phi_i}{2}} cos(\frac{\Phi_i}{2})|$.
 Furthermore, 
  $|\sum_{i=1}^M e^{j\frac{\Phi_i}{2}} cos(\frac{\Phi_i}{2})|\leq\sum_{i=1}^M |e^{j\frac{\Phi_i}{2}} cos(\frac{\Phi_i}{2}) |=\sum_{i=1}^M |cos(\frac{\Phi_i}{2})|\leq M$. 
Therefore, we have $M\leq \sum_{i=1}^M |cos(\frac{\Phi_i}{2})| \leq M$, hence $\sum_{i=1}^M |cos(\frac{\Phi_i}{2})|=M$. 

For $i=1,2,...,M$, $ |cos(\frac{\Phi_i}{2})|\leq 1$, thus
$\sum_{i=1}^M |cos(\frac{\Phi_i}{2})|= M$ holds only if $|cos(\frac{\Phi_i}{2})|= 1$. That is to say,  $cos(\frac{\Phi_i}{2}) = 1$ or $-1$, which implies that  $\frac{\Phi_i}{2}$ is a multiple of $\pi$, i.e., $\Phi_i$ is a multiple of $2\pi$.  
\end{proof}

\begin{remark} Define \textbf{Condition 1}: $k_i(L_p-L_q)$ is a multiple of $2\pi $, $i=1,2,...,M$.  If $L_p=L_q$, then  Condition 1 is satisfied without any restriction on frequency, which leads to the following conclusion: if the two point targets are on an ellipse with two focal points at $T_x$ and $R_x$, then their magnitudes are the maximum (=2M).
\end{remark}

From equations (\ref{eqn7}) and (\ref{eqn8}), we can write $A_p^*(A_p+A_q) =Ce^{i\theta}, A_q^*(A_p+A_q) =Ce^{-i\theta}$, where $C$ is the magnitude, $\theta$ and $-\theta$ are respectively the phases of the two point targets.  
Without loss of generality, in this paper we only consider $C>0$. 
Next, we list two special cases of the phases:\\
(1) if $\theta = 0$, then $-\theta =0$, i.e., the two phases are equal;\\
(2) if $\theta = \pi$, then $-\theta = -\pi$. Since $\pi$ and $-\pi$ are  considered the same in wrapped phase, the phases of the two point targets are equal.\\
In summary, if $\theta = 0$ or $\pi$, then the two point targets have same phase. 

Next, we prove that the case $\theta = \pi$ does not exist  in the two point targets reconstruction proposed in this section.
\begin{theorem}
Let the phases of the two reconstructed point targets be 
$\theta$ and $-\theta$, then $\theta$ can not be $\pi$.
\end{theorem}
\begin{proof}
If $\theta = \pi$, then $A_p^*(A_p+A_q)=-C,  A_q^*(A_p+A_q) = -C$. Together with (\ref{eqn7}) or (\ref{eqn8}), we have
\begin{equation}
\vspace{-2mm}
M+\sum_{i=1}^M cos(k_i(L_p-L_q)) =-C<0. \label{Contradiction}
\end{equation}
On the other hand, it is easy to see that $M+\sum_{i=1}^M cos(k_i(L_p-L_q))\geq 0$, thus equation (\ref{Contradiction}) is a contradiction, which implies that $\theta$ can not be $\pi$ in this reconstruction.
\end{proof}

From the above analysis, we know
that the two reconstructed point targets have opposite phase, and and they are the same only when  $\theta =0$. 
The question of when we can find a
condition under which $\theta =0$
is answered in the following theorem.
\begin{theorem} If the trivial case C=0 is excluded from consideration, then 
 $\theta=0$ if and only if $\sum_{i=1}^M sin(k_i(L_p-L_q))=0$.
\end{theorem}
\begin{proof}
$\Rightarrow$: If $\theta =0$, 
$A_p^*(A_p+A_q) =C$. Then by (\ref{eqn7}) we have $\sum_{i=1}^M sin(k_i(L_p-L_q))=0$.

$\Leftarrow$: If $\sum_{i=1}^M sin(k_i(L_p-L_q))=0$, then from (\ref{eqn7}) we have 
$A_p^*(A_p+A_q) = M+\sum_{i=1}^M cos(k_i(L_p-L_q))\geq 0$. In addition, $A_p^*(A_p+A_q) =Ce^{i\theta}$ and $C>0$, thus we obtain $\theta =0$.
\end{proof}

\begin{remark}
 $L_p=L_q$ obviously satisfies \textbf{Condition 2}: $\sum_{i=1}^M sin(k_i(L_p-L_q))=0$, which implies that regardless of the frequency, if the two point targets are on an ellipse with foci at $T_x$ and $R_x$, then their phases are equal to $0$.
\end{remark}

From Theorems 2-5, we know that the magnitudes of the two reconstructed point targets are the same, and only when Condition 1 is satisfied, they are maximum and equal to $2M$; the two phases are opposite, and only when Condition 2 is satisfied, they are equal $(=0)$.
Condition 1 is a special case of Condition 2, thus if Condition 1 holds, then Condition 2 holds as well, and if Condition 2 does not hold, then Condition 1 does not hold either. Therefore, if condition 1 holds, it can be obtained that the magnitudes and the phases of the two point targets are $2M$ and $0$, respectively, i.e., the reconstructed values are  $2M$. If Condition 2 holds but Condition 1 does not, then the two reconstructed point targets have same phase $0$ and same magnitude, but their magnitudes are not the maximum; that is to say, the two reconstructed values are positive , but less than $2M$. At last, if Condition 2 does not hold, their magnitudes are not maximum and their phases are not $0$.

\section{Numerical experiments}
In this section, we present four numerical experiments to verify the theorems developed in Sections II and III. Throughout this section, the imaging area is $[\rm -0.5 \ m, 0.5\ m]\times[\rm 0.5 \ m, 1 \ m]$, the transmitting antenna is used to illuminate the imaging area with multiple frequencies, and the receiving antenna is used to acquire the reflected signals. The radar frequencies are in the range of 56.5-64 GHz with step $0.25$ GHz, i.e., $M=30$. The transmitter and the receiver are located at $(\rm -0.2 \ m, 0.1\ m)$ and $(\rm 0.2 \ m, 0.1 \ m)$ respectively.  

\begin{example}
In this example, the point target is located at $(\rm 0.1\ m,0.7\ m)$. 
By calculation, the reconstructed value of the point target is $30$, that is to say, the magnitude is  equal to $M$  and the phase is $0$, which verifies Theorem 1. We show the magnitude of the reconstructed solution in Fig.~\ref{fig:SinglePointTarget}. The location of the point target is marked with green dot in the image. 
\begin{figure}[ht]
\centering
\vspace{-3mm}
\includegraphics[width=0.9\columnwidth]{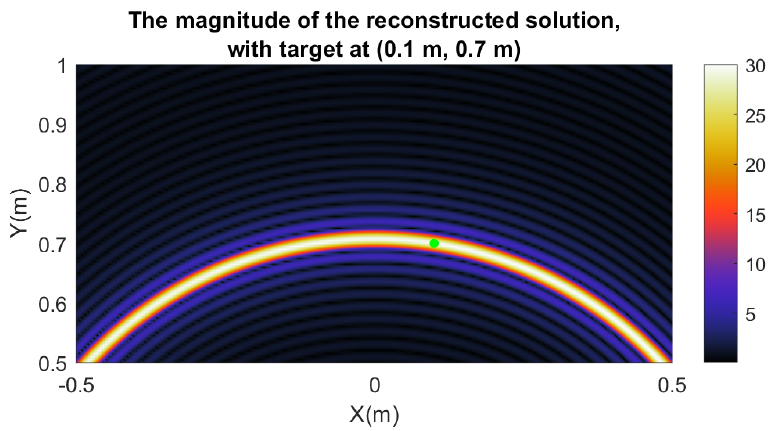}
\vspace{-3mm}
\caption{Reconstructed image of single point target at (0.1 m, 0.7 m), where the green dot represents the location of the target.\vspace{-3mm}} 
\label{fig:SinglePointTarget}
\end{figure}
\end{example}

\begin{example}
In this example, the two point targets are located at $(\rm -0.1\ m,0.6\ m)$ and $(\rm 0.3\ m, 0.8\ m)$. 
The reconstructed values obtained by the inverse crime proposed in Section III are $29.9332 - 0.7735i $ and $29.9332 + 0.7735i$, thus their magnitudes are the same, equal to $29.9432$, and their phases are opposite, equal to $-0.0258$ and $0.0258$ respectively. These results are obviously consistent with Theorem 2. The imaging result is displayed in Fig.~\ref{fig:TwoPointTargets}, where the green dots indicate the locations of the two point targets.
\begin{figure}[ht]
\centering
\vspace{-3mm}
\includegraphics[width=0.9\columnwidth]{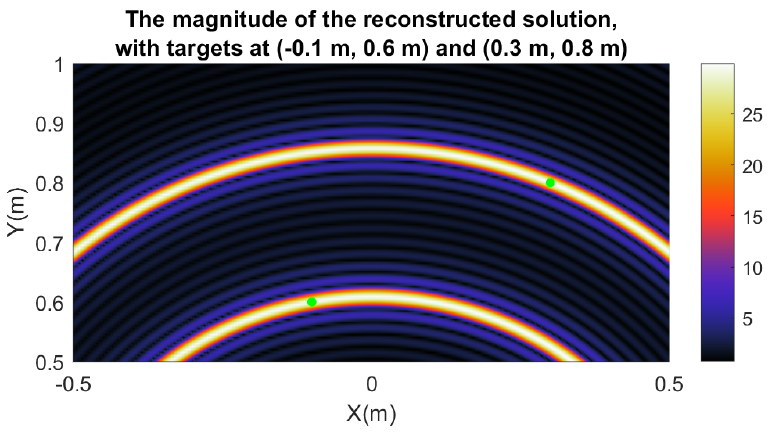}
\vspace{-3mm}
\caption{Reconstructed image of two point targets at (-0.1 m, 0.6 m) and (0.3~m, 0.8 m), where the green dots represent the locations of the targets.\vspace{-3mm}}
\label{fig:TwoPointTargets}
\end{figure}
\end{example}

\begin{example} 
In this example, we consider a special case where the two point targets locate on an  ellipse with foci at $T_x$ and $R_x$.  
The coordinates of the two point targets are $(\rm -0.25\ m, 0.75\ m)$ and $ (\rm 0.25\ m, 0.75\ m)$, thus they are symmetric and $L_p = L_q = 1.4425 \ \rm m$. It is apparent that the reconstructed values of the two point targets are both $60$, which equals to $2M$, therefore the results match Theorem 3 and Theorem 5. The imaging result can be found in
Fig.~\ref{fig:TwoPointTargetSymmetric}, where the two point targets are represented by green dots. 
\begin{figure}[ht]
\centering
\vspace{-3mm}
\includegraphics[width=0.9\columnwidth]{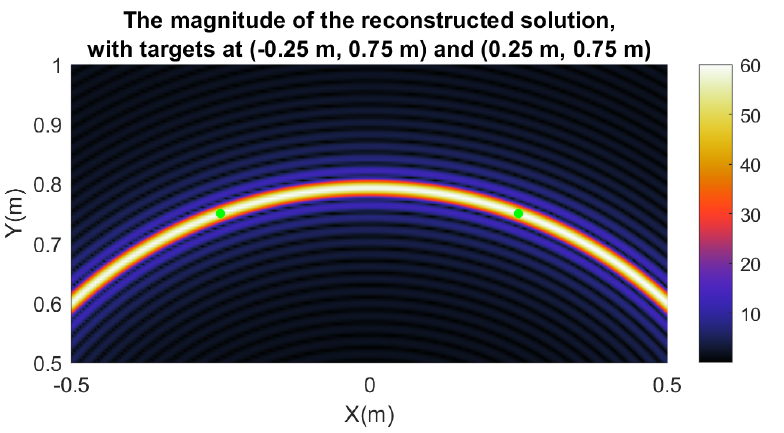}
\vspace{-3mm}
\caption{Reconstructed image of two symmetric point targets at (-0.25 m, 0.75 m) and (0.25 m, 0.75 m), where the green dots represent target locations.\vspace{-3mm}}
\label{fig:TwoPointTargetSymmetric}
\end{figure}
\end{example}

\begin{example} 
It is known from Theorem 1 and Theorem 2 that the sum of the phases of the point targets is always $0$ for the cases of one and two point targets. However, this conclusion on the phase does not hold for the cases of three and more point targets. For example, the reconstruction of three point targets at (-0.35 m, 0.9 m), (-0.25 m, 0.9 m) and (0.25 m, 0.9 m) yields that the sum of the phases of their reconstructed values is $-0.1576$. Fig.~\ref{fig:ThreePointTargets} shows the reconstructed image of this example. 
\begin{figure}[ht]
\centering
\vspace{-3mm}
\includegraphics[width=0.9\columnwidth]{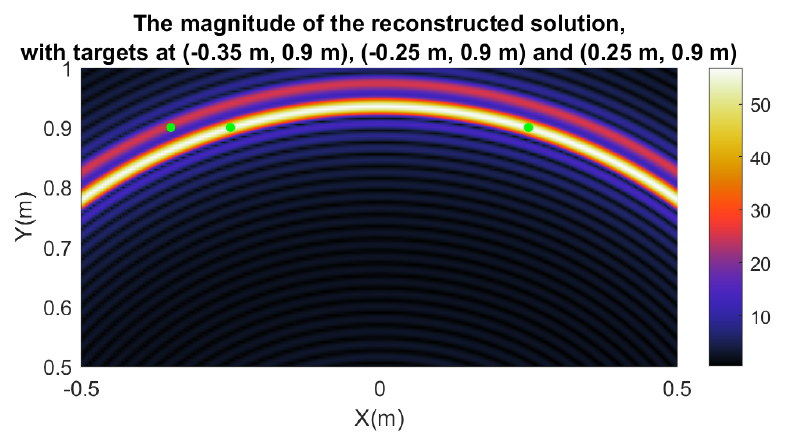}
\vspace{-3mm}
\caption{Reconstructed image of three point targets at (-0.35 m, 0.9 m), (-0.25 m, 0.9 m) and (0.25 m, 0.9 m), where the green dots represent target locations.\vspace{-3mm}}
\label{fig:ThreePointTargets}
\end{figure}
\end{example}

\section{Conclusions}
 The magnitude and phase of the reconstructed point targets in SAR imaging have been investigated via an inverse crime in this work. For single point target, it has been proved that the magnitude and phase are $M$ and $0$, respectively. For two point targets, it is demonstrated that their reconstructed values are conjugate of each other, i.e., with the same magnitude but opposite phase. Furthermore, the fact that the phase cannot be $\pi$ is addressed in Theorem 4. 
 Theorem 3 and Theorem 5
 propose Condition 1 under which the magnitude is maximum and Condition 2 under which the phase is $0$, respectively. The condition $L_p = L_q$, a special case of both Condition 1 and Condition 2, has been discussed on its effect on the magnitude and phase. Numerical experiments are also presented to verify the established theorems. Finally, for three or more point targets the sum of the phases might not be zero, although it is always zero for the cases of one and two point targets. This can be easily verified by numerical experiment.

\end{document}